\documentclass[11pt,3p]{article}

\usepackage{soul} \usepackage{amsmath}
\usepackage{amsfonts}
\usepackage{amssymb}
\usepackage{complexity}  
\usepackage{dot2texi}
\usepackage{todonotes}
\usepackage{hyperref}
\usepackage{url}
\usepackage{verbatim}
\usepackage{natbib}
\usepackage{enumitem}

\usepackage[ruled,vlined,linesnumbered]{algorithm2e}
\let\oldnl\nl
\newcommand{\nonl}{\renewcommand{\nl}{\let\nl\oldnl}}
\nonl \SetKwProg{Fn}{Procedure}{}{}

\usepackage[math]{cellspace}

\usepackage{cleveref}
\newcommand{\netroot}[1]{\rho_{#1}}

\usepackage{graphicx}

\usepackage{amssymb}

\usepackage{amsthm}

\newtheorem{proposition}{Proposition}[section]
\newtheorem{corollary}{Corollary}[section]
\newtheorem{theorem}{Theorem}[section]
 
 \usepackage{authblk}

\title{\large Solving Tree Containment Problem for Reticulation-visible Networks with Optimal Running Time}
\author{\normalsize Andreas D.M. Gunawan\thanks{a0054645@u.nus.edu}}
\affil{\small Department of Mathematics, National University of Singapore, Singapore 119076}

\date{\vspace{-5ex}}

\begin{document}

\maketitle

\begin{abstract}
Tree containment problem is a fundamental problem in phylogenetic study, as it is used to verify a network model.
It asks  whether a given network contain a subtree that resembles a binary tree.
The problem is NP-complete in general, even in the class of binary network. 
Recently, it was proven to be solvable in cubic time, and later in quadratic time for the class of general reticulation visible networks. 
In this paper, we further improve the time complexity into linear time.
\end{abstract}

\section{Introduction}

A binary tree is often used to model evolutionary history. The internal nodes of such tree represent speciation events (i.e. the emerging of a new species), and the leaves represent existing species.
However, a binary tree cannot explain reticulation events such as hybridization and horizontal gene transfer \citep{Chan_13_PNAS, marcussen2014ancient}.
This motivates researcher to develop a more general model, which is called phylogenetic networks.
In a phylogenetic networks, internal nodes of indegree more than one represent reticulation events, while other internal nodes represent speciation events. 

As a result of their experiment, biologists often obtained a binary tree that best explain the evolution of the gene/protein \citep{delsuc2005phylogenomics, ma2013reconstruction}.
Tree containment problem (TCP) is a problem that arise from verifying a given phylogenetic network model with the experimentally-derived binary tree. 
It asks whether there is a subtree in the phylogenetic model that is consistent with the binary tree.
However, the TCP is known to be NP-complete, even on the restricted class of binary phylogenetic network \citep{Kanj_08_TCS}.

In order to make the phylogenetic network model practical, much effort has been devoted to obtain classes of networks that are reasonably big, on which the TCP can be solved quickly.
One of the biggest known such class is the reticulation-visible networks.
The TCP for reticulation-visible networks was independently proven to be cubic-time solvable by \citet{BordewichSemple2015} and \citet{GDZ2016}. 
It is further improved into quadratic time in \citep{GDZ2016_2}, which is the journal version of \citep{GDZ2016}.

A certain decomposition theorem was introduced in \citep{GDZ2016} to solve the TCP.
The same decomposition is also used to produce a program to solve TCP for general network \citep{bingxin2016} and to obtain efficient program for computing Robinson-Foulds distance (RFD) \citep{APBC2017}.
The decomposition theorem enables us to decompose a network into several components, which can then be dissolved into a single leaf one by one, in a bottom-up manner.

In this paper, we further analyse the structure of a lowest component in a reticulation-visible network, which allows us to give an optimal algorithm with linear running time.

\section{Basic definitions and notations}

A phylogenetic network (or simply network) is a directed acyclic graph with exactly one root (nodes of indegree zero), and nodes other than the root have either exactly one incoming branch or exactly one outgoing branch. Node of indegree one is called tree node, and otherwise it is called reticulation node (or simply reticulation). For simplicity, we add an incoming branch with open end to the root, thereby making it a tree node. The set of leaves (tree nodes of outdegee zero) are labeled bijectively with a set of taxon, and represent the existing species under consideration.

For a given network $N$, 
$\mathcal{V}(N)$ denotes its set of nodes, 
$\mathcal{E}(N)$ its set of edges, 
$\mathcal{T}(N)$ its set of tree nodes (including root and leaves), 
$\mathcal{R}(N)$ its set of reticulations, and
$\mathcal{L}(N)$ its set of leaves.
The root of $N$ is denoted with $\rho_N$.

An edge is a reticulation edge if its head is a reticulation, and otherwise the edge is a tree edge. A path is a tree path if every edge in the path is a tree edge.

Node $u$ is a parent of node $v$ (or $v$ is the child of $u$) if $(u,v)$ is an edge in $N$. 
Two nodes are sibling if they share a common parent. 
For a node $v$, $\mbox{pr}_N(v)$, $\mbox{ch}_N(v)$, and $\mbox{sb}_N(v)$ denote the set of nodes (or the unique node if the set is a singleton) that is the parent, children, and sibling of $v$ in $N$.
In a more general context, node $u$ is above node $v$ (or $v$ is below $u$) if there is a path from $u$ to $v$. In such case, we also say that $u$ is an ancestor of $v$ and $v$ a descendant of $u$.
We always consider a node as below and above itself.
For a node $v$, $N[v]$ is defined as the subnetwork of $N$ induced by the nodes below $v$ and edges between them.

A phylogenetic network is binary, if every leaf is of indegree one and outdegree zero, while every other node has total degree of three. A phylogenetic tree is a binary phylogenetic network that has no reticulation.

For a set of nodes $V$, $N - V$ is the network with node set $\mathcal{V}(N)\backslash V$ and edge set $\{(u,v) \in \mathcal{E}(N) : u,v \notin V\}$. 
For a set of edges $E$, $N - E$ is the network with the same node set as $N$ and edge set $\mathcal{E}(N) \backslash E$. If the set $V$ or $E$ above contain only a single element $x$, we simply write the resulting network as $N - x$.

\subsection{Visibility property}

A node $u$ is the stable ancestor of (or is stable on) a node $v$ if any path from the root to $v$ pass through $u$ at some point.
A node is stable if it is the stable ancestor of some labeled leaf in $N$, and otherwise the node is unstable.
A reticulation-visible network is a network where every reticulation is stable. 

The following two propositions was proved in \citep{linearNS} as Proposition 2.1 and in \citep{GunawanZhang2015} as Lemma 2.3, which gives an insight on the structure around stable nodes.
\begin{proposition} \label{prop1}
The following facts hold:
\begin{enumerate}
     \item The parent of a stable tree vertex is always stable.
     \item A reticulation is stable if and only if its child is a stable tree vertex.
     \item If $u,v$ are stable ancestors of $w$, then either $u$ is above $v$ or vice versa.
\end{enumerate}
\end{proposition}

\begin{proposition}
\label{unstable_condition}
Let $u$ be a node in $N$, and $R$ be a set of reticulations below $u$ such that for each reticulation $r$ in $R$, either (i) $r$ is below another reticulation $r' \in R$, or (ii) there is a path from $\netroot{N}$ to $r$ that avoids $u$. Then,
$u$ is not stable ancestor any leaf $\ell$ below a reticulation in $R$. 
\end{proposition}
\noindent \Cref{unstable_condition} implies the following.
\begin{corollary}
A tree node with two reticulation children is unstable.
\end{corollary}

It is also worth noting that the number of nodes in a reticulation-visible network is bounded by the number of leaves. If $N$ is a binary reticulation-visible network with $n$ leaves, then there are at most $3(n-1)$ reticulations in $N$ (see \citep{BordewichSemple2015,GunawanZhang2015} and *cite nearly stable paper*). This implies that $|\mathcal{V}(N)|$ and $|\mathcal{E}(N)|$ are both $O(n)$ in size, which allows us to bound time complexity of algorithms in the number of leaves.

\begin{figure}
    \centering
    \includegraphics[width = 0.9\textwidth]{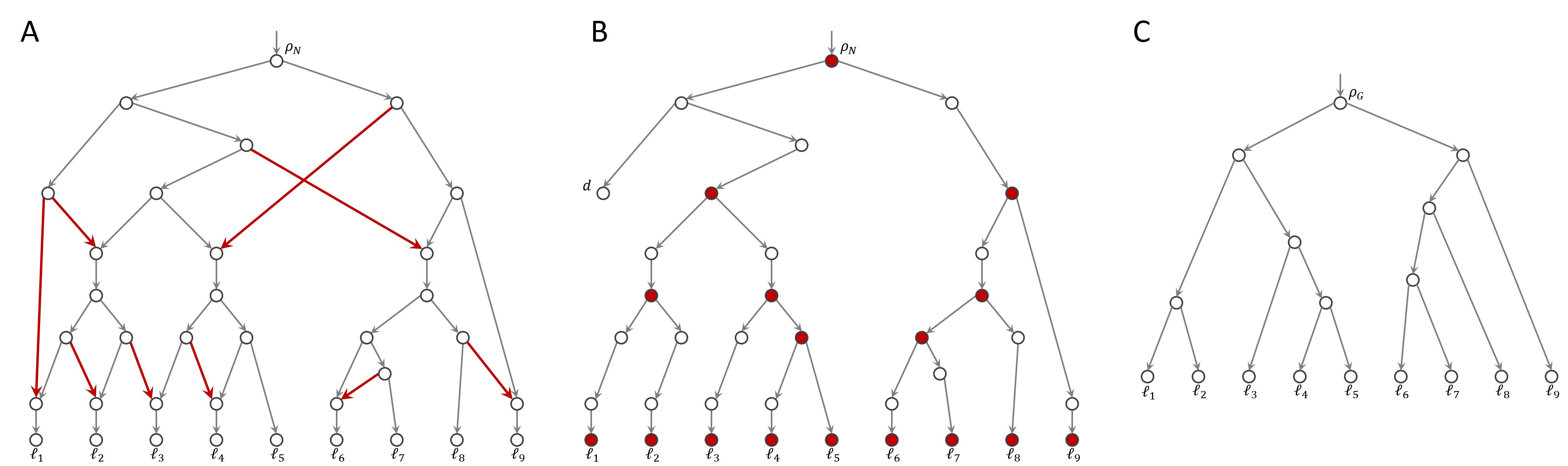}
    \caption{(\textbf{A}) An example of a reticulation-visible network with nine reticulations and nine leaves. 
    (\textbf{B}) A spanning subtree obtained by removing the highlighted reticulation branches in (A). After the edge removal, node $d$ becomes a dummy leaf. Highlighted nodes are the nodes preserved if we repeatedly contract incoming edges to dummy leaves and nodes of indegree and outdegree one. Note that there is a node of total degree three between $d$ and $\rho_N$ which need to be removed.
    (\textbf{C}) The resulting binary tree after edge contractions, which implies that this tree is displayed in the network in (A). }
    \label{Fig_TCP_example}
\end{figure}

\subsection{Tree containment problem}

Let $N$ be a binary phylogenetic network, and let $T$ be a subtree of $N$ containing $\rho_N$ and $\mathcal{L}(N)$. 
Contracting an edge $(u,v)$ from $T$ means we remove node $v$ and all edges incident to it, and modify the neighbourhood of $u$ as follows:
\begin{enumerate}
     \item For every edge $(w,v)$ that is removed, we add the edge $(w,u)$ if it doesn't exist; and
    \item For every edge $(v,x)$ that is removed, we add the edge $(u,x)$ if it doesn't exist.
\end{enumerate}
A leaf in $T$ that is an internal node in $N$ is not labeled with any taxa, and such leaf is called a dummy leaf. 
Contracting the incoming edge of a dummy leaf simply means we remove the dummy leaf.
The tree $T$ is said to be a subdivision of a phylogenetic tree $G$ if we can obtain $G$ by repeatedly contract incoming edges of dummy leaves and nodes of indegree and outdegree one from $T$ until there is no such nodes anymore.

The tree containment problem (TCP) for a network $N$ and a phylogenetic tree $G$ over the same set of taxa, is asking whether there exists a (spanning) subtree $T$ of $N$ that is a subdivision of $G$.

\subsection{Comparing two trees}

Let $T$ be a tree, and let $L\subseteq \mathcal{L}(T)$. We can pre-process the tree $T$ in $O(|\mathcal{V}(T)|)$, such that upon any given set of leaves $L$,
we can obtain in $O(|L|)$ time a binary tree $T'$ that is a subdivision of a subtree of $T$, such that $T'$ has leaf set $L \cap \mathcal{L}(T)$ (For example, see Section 8 of \citep{cole2000n}). 

The following subroutine algorithm to check whether a tree $T$ contains another tree $G$ is often used for the rest of the paper.
\begin{center}
\footnotesize
\begin{tabular}{l}
\hline
{\sc IsSubtree($T$, $G$) }\\
\\
{\bf Input}: A tree $T$ and a binary tree $G$\\
{\bf Output}: "YES" if there is a subtree of $T$ that displays $G$ and "NO" otherwise\\
\\
1. Traverse $T$ to find the set of leaves $\mathcal{L}(T)$;\\
2. Traverse $G$ to find the set of leaves $\mathcal{L}(G)$,\\
\hspace*{1em} but terminate and return "NO" if we visit more than $|\mathcal{V}(T)|$ nodes;\\
3. If ($\mathcal{L}(G) \not \subseteq \mathcal{L}(T)$) \{ return "NO"\};\\
4.  Else \{Find a binary subtree $T'$ satisfying $\mathcal{L}(T') = \mathcal{L}(G)$,\\
\hspace*{2em} such that there is a subtree of $T$ that is a subdivision of $T'$ as in \citep{cole2000n}.\}\\
5. If $T'$ is isomorphic with $G$, return "YES"; else return "NO".
\\
\hline

\end{tabular}
\end{center}

\section{A decomposition theorem}

For the rest of the paper, we assume that $N$ is a binary reticulation-visible network.

The following subsection relies on a decomposition theorem that was established in \citep{GDZ2016_2}. Readers who are interested for the complete proof of the discussion should refer to the paper.

Removing every reticulation from $N$ generates a forest $N - \mathcal{R}(N)$, where every node in the forest must be a tree node in $N$. 
Each maximal connected component in the forest consists of tree nodes in $N$, and is called a tree component of $N$.
Let $C_0, C_1, C_2, \ldots, C_q$ denote the tree components of $N$, where $C_0$ denotes the special tree component rooted at $\rho_N$.

Let $\rho_i$ denote the root of tree component $C_i$ for all $i = 1, 2, \ldots, q$, and set $\rho_0 = \rho_N$ for convenience. 
A tree component root is either the network root $\rho_N$ or the child of a reticulation.
As $N$ is reticulation-visible, a tree component root is always stable according to \Cref{prop1}.
A tree component is big if it contains at least two nodes, and otherwise it is called a leaf component.
In a reticulation-visible network, a tree component is either a leaf component or a big tree component.

A tree component $C_j$ is below another component $C_i$ if  $\rho_j$ is below $\rho_i$.
We can order the component roots such that $\rho_j$ is below $\rho_i$ only if $i < j$, for instance via breadth-first search on $N$. 

Let $j$ be the largest index such that $C_j$ is a big tree component. 
Every tree component below $C_j$ are simply leaf component, and therefore $C_j$ is called the lowest big tree component in $N$. 
Every leaf of $N$ below $\rho_j$ is either included in $C_j$ or is the child of a reticulation $r$, where $r$ has at least one parent in $C_j$.

Let $s$ be a node in the lowest component $C_j$.
We classify the leaves below $s$ into three types.
A leaf is of type-1 (with respect to $s$) if there is a tree path from $s$ to the leaf.
Leaf whose parent is a reticulation are called type-2 if both parents of the reticulation are below $s$, and otherwise it is type-3 leaf.
It is not hard to see that $s$ is stable only on type-1 and type-2 leaves.
Let $L_1(s), L_2(s), $ and $L_3(s)$ denote the set of leaves of type-1, 2, and 3 with respect to $s$.

\begin{figure}
    \centering
    \includegraphics[width = 0.4\textwidth]{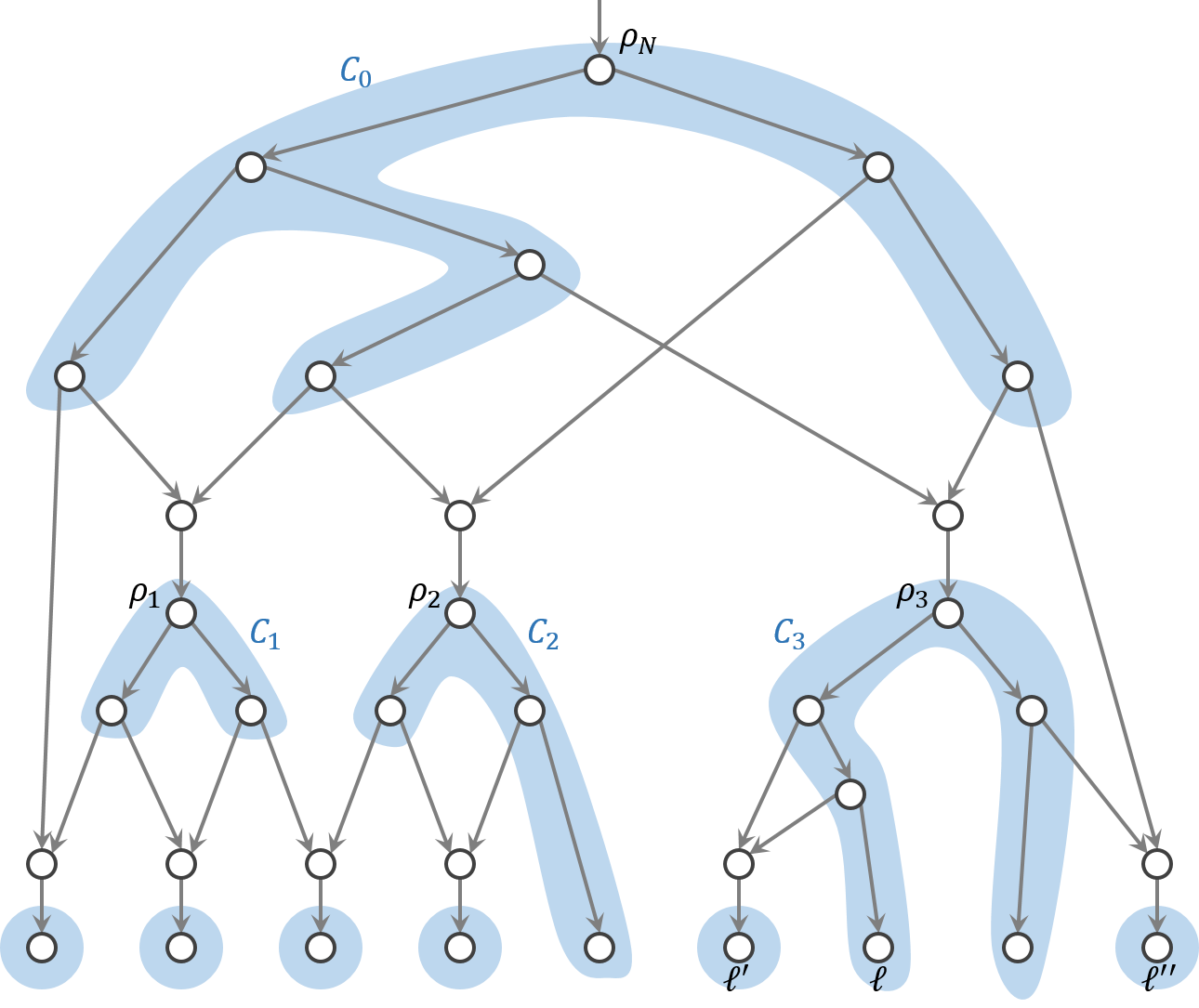}
    \caption{The tree components of a reticulation visible network. There are four big tree components, namely $C_0$, $C_1$, $C_2$, and $C_3$, and six other leaf components. The leaves $\ell, \ell', $ and $\ell''$ are respectively of type-1, 2, and 3 with respect to $\rho_3$}
    \label{tree_decomp}
\end{figure}

\subsection*{Overview of the linear-time algorithm}

By using the decomposition theorem, we can use divide-and-conquer approach to solve the tree containment problem.
First, we pre-process the input network $N$ to decompose it into its tree components. 
We then observe the lowest tree component, dissolve it into a single leaf, and recurse on the next lowest tree component. 

Dissolving the lowest component efficiently is not trivial.
In the next section, we show that there is a set $S$ of stable nodes in $C_j$, such that if a node $s \in S$ is a lowest node in $S$, then the subnetwork of $N$ below the children of $s$ are merely two trees. 
By utilizing the fact that checking whether a tree is inside another tree is easy, we can cut several reticulation branches below $s$, contract $N[s]$ into a single leaf, and repeat this for all nodes in $S$ to eventually the lowest tree component is dissolved into a leaf component.

\section{Node with special properties and the structure below it}

Here, let $C_j$ be a lowest tree component in $N$, and suppose $s$ is a node in $C_j$ satisfying the following properties:
\begin{itemize}
    \item[I.] $s$ is a stable tree vertex;
    \item[II.] $s$ has two children, namely $s'$ and $s''$; and
    \item[III.] $N[s']$ and $N[s'']$ are both trees.
\end{itemize}
We will prove that we can contract $N[s]$ (probably along with a subtree of $G$) into a single leaf, such that the resulting network displays the resulting binary tree if and only if $N$ displays $G$.

As $s$ is a stable tree vertex, it is stable on either a leaf in $L_1(s)$ or in $L_2(s)$. 
We further consider three possible cases for $s$.

\subsection{Case C1: there are two edge-disjoint paths from $s$ to two leaves}

As the paths are edge-disjoint, one must pass through $s'$ and ends at a leaf $\ell'$, while the other pass through $s''$ and ends at $\ell''$.
Let $t$ be the lowest common ancestor of $\ell'$ and $\ell''$ in $G$, and let $t', t''$ be the children of $t$ on the path from $t$ to $\ell'$ and $\ell''$ in $G$, respectively. 

\begin{proposition} \label{prop_c1}
If $C1$ holds, then $N$ displays $G$ if and only if the followings hold.
\begin{itemize}
    \item[(i)]  $N[s']$ displays $G[t']$ and $N[s'']$ displays $G[t'']$.    
    \item[(ii)] $L_1(s) \cup L_2(s) \subseteq \mathcal{L}(G[t]) \subseteq \mathcal{L}(N[s])$
    \item[(iii)] Let 
    \begin{equation}
    \label{eqn_X}    
    X = \mathcal{V}(N[s]) \backslash \{\ell ,\mbox{pr}_N(\ell): \ell \in \mathcal{L}(N[s])\backslash \mathcal{L}(G[t])\}.
    \end{equation}
    If $N' = N - (X \backslash \{s\}) $ and $G' = G - (\mathcal{V}(G[t]) \backslash \{t\})$ and we label $s$ of $N'$ and $t$ of $G'$ with a new taxa, then $N'$ displays $G'$.
\end{itemize}
\end{proposition}
\begin{proof}

We first prove the sufficiency, and assume that conditions (i), (ii), and (iii) holds.
According to (i), there are subtrees $T_1$ and $T_2$ of $N[s']$ and $N[s'']$ that are subdivisions of $G[t']$ and $G[t'']$, respectively. 
The set $X$ in \Cref{eqn_X} represent the nodes in $N[s]$, excluding leaves that are not below $t$ in $G$ and their reticulation parents.
Assumption (2) implies that the leaves excluded from $X$ are in $L_3(s)$, therefore ensuring that they are reachable from $\rho_N$ in $N'$.
Finally, assumption (3) implies there is a subtree $T_3$ of $N'$ that displays $G'$. Combining $T_1$, $T_2$, $T_3$, and the edges $\{ (s,s'), (s,s'') \}$ yields a tree $T$ of $N$  that is a subdivision of $G$.

Next, we prove the necessity. Assume that $N$ displays $G$, so there is a subtree $T$ of $N$ that is a subdivision of $G$. We can further assume that $T$ does not contain any dummy leaf. Note that any stable node (in particular $s$) must be in $T$, as otherwise there is a leaf of $N$ not reachable from the root of $T$.

To prove (i), we consider $T[s]$. The two tree edge-disjoint paths from $s$ to $\ell'$ and $\ell''$ in $N[s]$ must be included in $T[s]$ because every path from $\rho_N$ to $\ell'$ (resp. $\ell''$) must contain the path from $s$ to $\ell'$ (resp. $\ell''$). Therefore, $s$ is the lowest common ancestor of $\ell'$ and $\ell''$ in $T$, and $T[s]$ must be a subdivision of $G[t]$. 
$T[s']$ must contain the leaf $\ell'$, while $T[s'']$ must contain the leaf $\ell''$. Therefore $T[s']$ (resp. $T[s'']$) is a subdivision of $G[t']$ (resp. $G[t'']$).

Condition (ii) is an immediate result from (i); $s$ is a stable ancestor of every leaf in $L_1(s) \cup L_2(s)$, and so $T[s]$ must contain all of these leaves. As $T[s]$ is a subdivision of $G[t]$, then the left part of the equation holds. The right part of the equation holds simply because $T[s] \subseteq N[s]$.

To prove (iii), we consider $T' = T - (\mathcal{V}(T[s]) \backslash \{s\})$. 
The nodes in $\mathcal{V}(T[s])$ is a subset of $X$ because of the fact that $T[s]$ is a subdivision of $G[t]$ (therefore does not contain a leaf $\ell^* \notin \mathcal{L}(G[t])$) and the assumption that $T$ does not contain any dummy leaf (therefore does not contain $\mbox{pr}_N(\ell^*)$).
Hence $T'$ is a subtree of $N'$, and is the evidence that $N'$ displays $G'$.
\end{proof}

Therefore, we can use the following subroutine algorithm to dissolve $N[s]$ whenever case C1 is found. From now on, $lca_T(u, v)$ denotes the lowest common ancestor of nodes $u$ and $v$ in a tree $T$.

\begin{center}
\footnotesize

\begin{tabular}{l}
\hline
{\sc [$N'$, $G'$] = Dissolve\_C1 ($N$, $s$, $G$) }\\
\\
{\bf Input}: A binary phylogenetic tree $G$ and a reticulation-visible network $N$. \\
\hspace*{1em}  The node $s$ in $N$ satisfies properties I, II, III, and C1.\\
 {\bf Output}: "NO" if $N$ does not display $G$, otherwise output $N'$ and $G'$ as in\\
 \hspace*{1em}\Cref{prop_c1}.\\
\\

1. Set $s'$ and $s''$ to be the children of $s$;\\
\hspace*{1em} Find leaves $\ell', \ell''$ so there are tree paths from $s'$ to $\ell'$ and from $s''$ to $\ell''$;\\
2. Set $t = \mbox{lca}_G(\ell', \ell'')$;\\
\hspace*{1em} Set $t'$ to be the child of $t$ that is an ancestor of $\ell'$ and $t''$ to be the sibling of $t'$.\\
3. If ($L_1(s) \cup L_2(s) \not \subseteq \mathcal{L}(G[t])$ or  $\mathcal{L}(G[t])\not \subseteq \mathcal{L}(N[s])$) \{stop and return "NO"\}\\
4. Check whether $N[s']$ displays $G[t']$ as follows:\\
\hspace*{1em} If ({\sc IsSubtree($N[s'], G[t']$)} returns no) \{stop and return "NO" \}\\
\hspace*{1em} Else \{ \\
\hspace*{2.5em} Set $X' = \mathcal{V}(N[s']) \backslash \{\ell ,\mbox{pr}_N(\ell) : \ell \in \mathcal{L}(N[s'])\backslash \mathcal{L}(G[t'])\}$;\\
\hspace*{2.5em} Set $N = N - X'$; \}\\
5. Check whether $N[s'']$ displays $G[t'']$;  ~~\%Similar as step 4, with $s'', t''$ replacing $s',t'$\\
6. Set $G = G - (\mathcal{V}(G[t]) \backslash \{t\})$;\\
\hspace*{1em} Label $s$ in $N$ and $t$ in $G$ with new taxa,\\
\hspace*{1em} Output the resulting network as $N'$ as $G'$.

\\
\hline

\end{tabular}
\end{center}

The correctness of the algorithm follows from \Cref{prop_c1}. 
Note that the set $X'$ from step 4 comprises of nodes in $N[s']$ other than leaves that does not belong to $G[t']$ and their parent. In step 5, another set $X''$ is defined, and the two sets satisfy $X' \cup X'' = X \backslash \{s\}$, where $X$ is the set defined in \Cref{eqn_X}.

\subsection{Case C2: there is a tree path from \texorpdfstring{$s$}{Lg} to a leaf}

We remark first that case C1 is a special case of C2. The obstacle in case C2 is that it is not easy to pinpoint the node $t$ in $G$ that should correspond to $s$, because now we only have one tree path from $s$ to a leaf.

Assume that there is a tree path from $s$ to a leaf $\ell$. 
We set $u_1 = \ell$ in $N$, and recursively define $u_{i+1} = \mbox{pr}_N(u_{i})$ until it reaches $u_{k+1} = s$.
We also set $v_1 = \ell$ in $G$, and recursively define $v_{i+1} = \mbox{pr}_G(v_{i})$ whenever needed.
We further define $u_i'$ (resp. $v_i'$) to denote the sibling of $u_i$ (resp. $v_i$) whenever possible.
Let $j_1 = 1$, and recursively define $j_i, i>1$  to be the smallest integer satisfying $j_{i} > j_{i-1}$ and $N[u_{j_i-1}']$ displays $G[v_{i-1}']$ (or, equivalently, $N[u_{j_i}]$ displays $G[v_i]$) whenever possible. Let $l$ be the highest index such that $j_l$ is defined, and set $t = v_l$.

Intuitively, this means that we greedily choose a 'lowest' subtree in $N$ for constructing a subdivision of $G$ whenever possible.
Similar as in case C1, $N$ displays $G$ if and only if $N[s]$ displays $G[t]$ and $N - (X\backslash \{s\})$ displays $G- (X\backslash \{s\})$. 
An illustration for the node labeling and the process for finding the indexes $j_i$s are shown in \Cref{case_C2}.

\begin{figure}
    \centering
    \includegraphics[width = \textwidth]{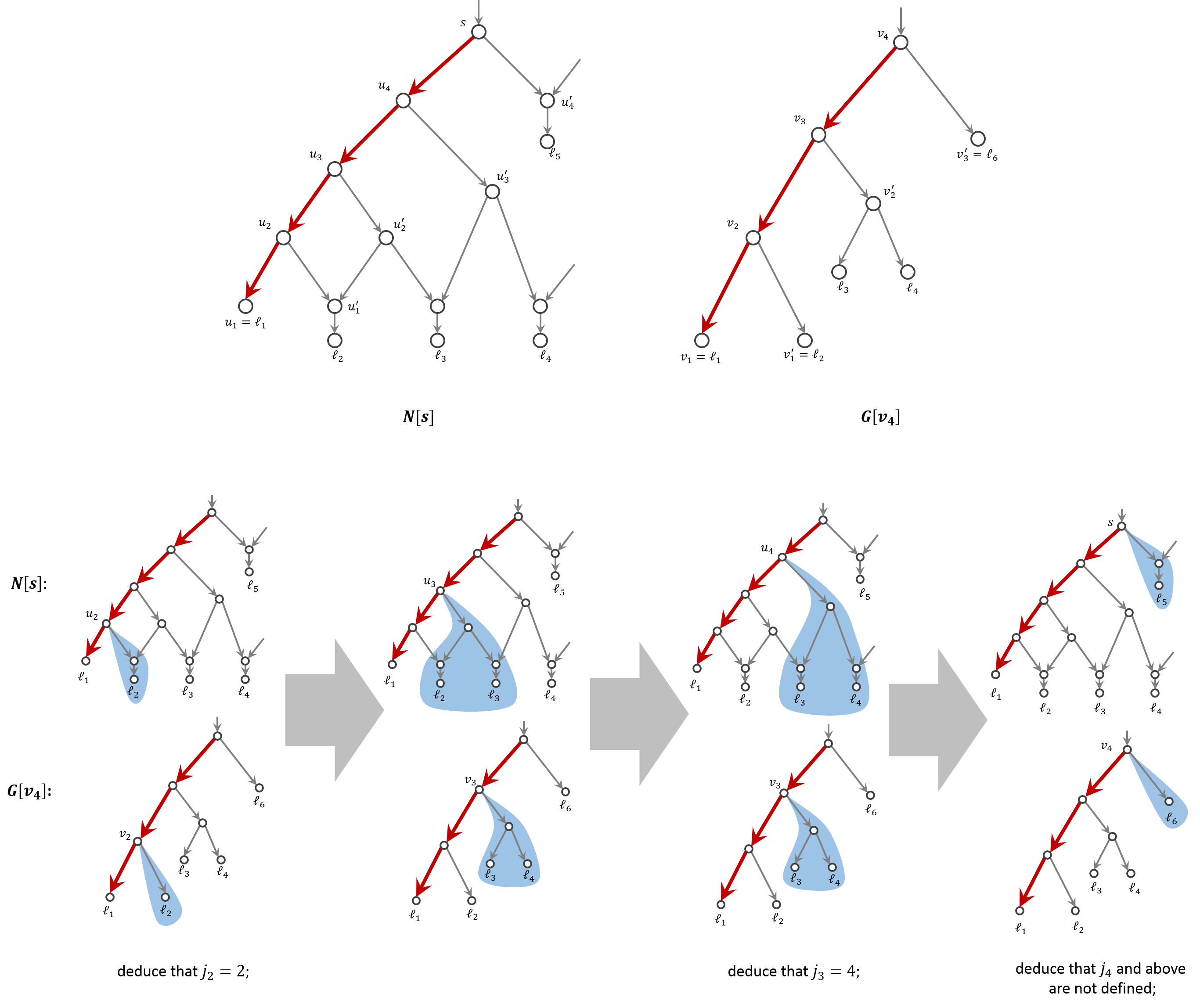}
    \caption{
    \textbf{Top}: an illustration for the node labeling in case C2, where we choose $u_1 = \ell_1$ in $N$ as the leaf to which there is a tree path from $s$. $u_{i+1}$ and $u_i'$ is defined as the parent and sibling of $u_{i}$, respectively. The indexes $v_i$ are defined in a similar manned in $T$.
    \textbf{Bottom}: a series of illustration showing the process of finding the indexes $j_i$s.  Initially, set $j_1 = 1$. We then compare the subnetwork of $N$ branching off $u_2$ with the subtree of $T$ branching off $v_2$. As the subtree is displayed in the subnetwork, we deduce that $j_2 = 2$.
    Next, we compare the subnetwork of $N$ branching off $u_3$ with the subtree of $T$ branching off $v_3$. As the subtree is not displayed, we moved on by comparing the subtree with the subnetwork of $N$ branching off $u_4$, and we can deduce that $j_3 = 4$ as $N[u_3']$ displays $G[v_2']$. 
    Finally, $j_4$ is not defined (because the leaf $\ell_6$ is not found in $N[s]$), and therefore we set $t = v_3$.}
    \label{case_C2}
\end{figure}

We formally state this in the following proposition.

\begin{proposition}\label{prop_c2}
 $N$ displays $G$ if and only if the following holds:
 \begin{itemize}
     \item[(i)] $N[u_{j_i}']$ displays $G[v_i']$ for every $i = 1,2, \ldots, l$,
     \item[(ii)] $L_1(u_{j_i}) \cup L_2(u_{j_i}) \subseteq \mathcal{L}(G[v_i]) \subseteq \mathcal{L}(N[u_{j_i}])$ for every $i = 1,2, \ldots, l$, and
     \item[(iii)] Let $X$ be defined as in \Cref{eqn_X}.
     If $N' = N - (X \backslash \{s\})$ and $G' = G - (X \backslash \{t\})$ and we label $s$ of $N'$ and $t$ of $G'$ with a new taxa, then $N'$ displays $G'$.
 \end{itemize}
\end{proposition}
\begin{proof}
We first prove the sufficiency. Assume conditions (i), (ii), and (iii) holds, we need to prove that $N$ displays $G$.

By condition (i), there is a subtree $T_i$ of $N[u_{j_i}']$ that is a subdivision of $G[v_i']$. We can further assume that each tree $T_i$ has no dummy leaf and is rooted at $u_{j_i}'$. 
Note that $N[s']$ and $N[s'']$ are trees, so for every distinct pair $a,b \in \{1, 2, \ldots, l\}$, the trees $T_a, T_b$ are disjoint except perhaps on some leaves in $L_2(s)$ and their parents, i.e.
$$\mathcal{V}(T_a) \cap \mathcal{V}(T_b) \subseteq \{\ell, \mbox{pr}_N(\ell) : \ell \in L_2(s) \}.$$
$T_a$ and $T_b$ have disjoint leaf sets, because the trees are subdivisions of disjoint subtrees of $G$. Moreover, $T_a$ and $T_b$ do not contain any dummy leaf, so at most one of them contain the parent of a leaf in $L_2(s)$.
Hence the trees $\{T_i\}$ are pairwise node-disjoint.
We can then construct a tree $T^*$ by combining edges of $T_i$s, the edges $\{(u_{j_i + 1}, u_{j_i}'), i = 1,2,\ldots, l\}$, and the tree path from $s$ to $\ell$. $T^*$ is a subdivision of $G[t]$.

By condition (iii), there is  a subtree $T'$ of $N'$ that is a subdvision of $G'$.
By a similar reasoning as in \Cref{prop_c1}, 
we then combine $T^*$ and $T'$ into a subtree $T$ of $N$ that is a subdivision of $G$. This completes the proof for the sufficiency.

To prove the necessity, we assume that $N$ displays $G$, so there is a subtree $T$ of $N$ that is a subdivision of $G$. We can further assume that $T$ has no dummy leaf. 

Condition (i) is immediate by how we define the indexes $j_i$'s. 

Neext, we prove condition (ii). The node $u_{j_i}$ is stable on the leaves in $L_1(u_{j_i}) \cup L_2(u_{j_i})$, and therefore the leaves must be below $u_{j_i}$ in any subtree of $N$ that subdivides $G$, which gives us the left part of the inequality in (ii). The right part of the inequality can be proved by induction, as $\mathcal{L}(G[v_1]) = \mathcal{L}(N[u_1])$ and the recursive relation
$$\mathcal{L}(G[v_{i}]) = \mathcal{L}(G[v_{i-1}]) \cup \mathcal{L}(G[v_{i-1}']) \subseteq \mathcal{L}(N[u_{j_{i-1}}]) \cup \mathcal{L}(N[u_{j_i-1}']) \subseteq \mathcal{L}(N[u_{j_{i}}]).$$

Finally, we prove condition (iii). As tree $T$ is a subdivision of $G$, there is a node in $T$, say $w$, that correspond to node $t$ in $G$, such that $T[w]$ is a subdivision of $G[t]$.
As the node $s$ is stable, it must be in $T$. Moreover, $s$ and $w$ are both in the path from $\rho_N$ to $\ell$ in $T$. We consider two cases.

First, assume that $w$ is strictly below $s$. 
Let $P_T[s,w]$ be the path from $s$ to $w$ in $T$.
If there is an edge $e \in \mathcal{E}(T) \backslash \mathcal{E}(P_T[s,w])$ whose tail is in $\mathcal{V}(P_T[s,w]) \backslash \{w\}$, then $\mathcal{L}(T[w]) \subsetneq \mathcal{L}(T[s])$ as $T$ has no dummy leaf. 
This implies that $N[s]$ displays $G[\mbox{pr}_G(t)]$, contradicting the maximality of $t$. 
Therefore, no such edge may exist, and so $\mathcal{L}(T[s]) = \mathcal{L}(T[w]) = \mathcal{L}(G[t])$.
It is not hard to see that $T - (\mathcal{V}(T[s]) \backslash \{s\})$ is then a subtree of $N'$ that displays $G'$. 

Next, assume that $w$ is strictly above or equal to $s$. Then the tree $T' := T - \mathcal{V}(T[w] \backslash \{w\} + P_T[w,s])$ is a subtree of $N'$ that displays $G'$. 
This completes the proof.
\end{proof}

Algorithm {\sc Dissolve\_C2} can be called to dissolve $N[s]$ if case C2 holds.

\begin{center}
\footnotesize
\begin{tabular}{l}
\hline
{\sc [$N'$, $G'$, $l$] = Dissolve\_C2 ($N$, $s$, $G$) }\\
\\
{\bf Input}: A binary phylogenetic tree $G$ and a reticulation-visible network $N$.  \\
\hspace*{1em} The node $s$ satisfies properties I, II, III, and C2.\\
{\bf Output}: "NO" if $N$ does not display $G$, otherwise output  $N'$ and $G'$ as in\\ 
\hspace*{1em} \Cref{prop_c2}. The index $l$ is an optional output.\\
\\

1. Find a leaf $\ell$ such that there is a tree path from $s$ to $\ell$;\\
\hspace*{1em} Set $u_1 = \ell$ and recursively define $u_{i+1} = \mbox{pr}_N(u_i)$ until $u_{k+1} = s$;\\
\hspace*{1em} Set $u_i' = \mbox{sb}_N(u_i)$ for every $i = 1, 2, \ldots k$;\\
2. Set $v_1 = \ell, v_2 = \mbox{pr}_G(v_1),$ and $v_1' = \mbox{sb}_G(v_1)$;\\
3. Set $l = 1$, $\alpha = 2$ and $\beta = 2$;   ~~~~\% $\alpha$ iterates on $N$, $\beta$ iterates on $G$\\
4. For ($\alpha = 2 : k$), do \{\\
\hspace*{1em} 4.0 set DISP = 0; ~~~~\% DISP is a flag showing whether new subtree of $G$ is found;\\
\hspace*{1em} 4.1 Traverse $N[u_\alpha']$ and find $L_1(u_\alpha')$;\\

\hspace*{1em} 4.2 
If ($u_\alpha'$ is a reticulation) \{\\
\hspace*{4em} If ( $\mathcal{V}(G[v_\beta']) = \{\mbox{ch}_N(u_\alpha')$\} )\{\\
\hspace*{5.5em} Set DISP = 1, $N = N - \{u_\alpha, u_\alpha', \mbox{ch}_N(u_\alpha')\}$, $G = G - \{v_\beta, v_\beta'\}$;\\
\hspace*{5.5em} Label $u_{\alpha+1}, v_{\beta +1}$ with the same taxon as $u_\alpha$; \}\\
\hspace*{4em} Else \{ \\
\hspace*{5.5em} Set $N = N - (u_\alpha+1, u_\alpha') - u_\alpha$;\\
\hspace*{5.5em} Label $u_{\alpha+1}$ with the same taxon as $u_\alpha$;\}\}\\

\hspace*{1em}4.3
ElseIf $(L_1(u_\alpha') \neq \emptyset)$ \{\\
\hspace*{4em} If ( {\sc Dissolve\_C1($N,u_\alpha,G$)} = "NO" ) \{ stop and return "NO" \};\\
\hspace*{4em} Else \{ [$N, G$] = {\sc Dissolve\_C1($N,u_\alpha,G$)} and DISP = 1\}\} \% end if\\

\hspace*{1em} 4.4
Else \{\\
\hspace*{4em} If ( {\sc IsSubtree($N[u_\alpha'], G[v_\beta']$)} = "YES" )\{\\
\hspace*{5.5em} Set $X_\alpha = \mathcal{V}(N[u_\alpha']) \backslash \{\ell ,\mbox{pr}_N(\ell) : \ell \in \mathcal{L}(N[u_\alpha'])\backslash \mathcal{L}(G[v_\beta'])$\};\\
\hspace*{5.5em} Set DISP = 1, $N = N - X_\alpha - u_\alpha$, $G = G - \mathcal{V}(G[v_\beta']) - v_\beta)$;\\
\hspace*{5.5em} Label $u_{\alpha+1}, v_{\beta +1}$ with the same taxon as $u_\alpha$; \}\\
\hspace*{4em} Else \{ \\
\hspace*{5.5em} Set $Y_\alpha = \mathcal{V}(N[u_\alpha']) \backslash \{\ell ,\mbox{pr}_N(\ell) : \ell \in
\mathcal{L}(N[u_\alpha']) \}$;\\
\hspace*{5.5em} Set $N = N - Y_\alpha - u_\alpha$;\\
\hspace*{5.5em} Label $u_{\alpha+1}$ with the same taxon as $u_\alpha$;\}\}\% end else\\

\hspace*{1em} 4.5 If (DISP = 1) \{set $l = \beta$, $\beta = \beta +1$, and find $v_\beta, v_\beta '$\};\} \% end for\\
5. Output the resulting network and tree as $N'$ and $G'$; Output $l$ if queried.

\\
\hline

\end{tabular}
\end{center}

In step 4.2, if $u_\alpha'$ is a reticulation, then its child must be a leaf. Then $N[u_{\alpha+1}]$ displays $G[v{\beta+1}]$ if and only if $v_\beta'$ is precisely the leaf child of $u_\alpha'$. 

If $u_\alpha'$ is a tree node and there is a tree path from it to a leaf, we may call the subalgorithm {\sc Dissolve\_C1}.
If $L_1(u_\alpha')$ is empty, that means every leaf below $u_\alpha'$ is a leaf in $L_3(u_\alpha')$. This is because $N[u_\alpha']$ is either a subtree of $N[u_k]$ if $\alpha \neq k$ or equal to $N[u_k']$ if $\alpha = k$, both of which are trees (property III), which in turn implies that $L_2(u_\alpha') = \emptyset$.
If $N[u_\alpha']$ displays $G[v_\beta]$, then we found a new index $j_i$, and we update the network accordingly.
Otherwise, we remove every reticulation edge in $N[u_\alpha']$, and contract $N[u_{\alpha+1}]$ into a new leaf labeled with the taxon of $u_\alpha$.

\subsection{Case C3: \texorpdfstring{$s$}{Lg}  has two unstable children.}
 
By property I, $s$ is a stable node, therefore $L_1(s) \cup L_2(s) \neq \emptyset$. But property C3 implies that $L_1(s)$ is empty, as otherwise at least one of the children of $s$ must be stable. 
Thus, we can assume that $\ell$ is a leaf in $L_2(s)$.
We let $e_1, e_2$ be the incoming edges of $\mbox{pr}_N(\ell)$, and let $N_1 = N-e_1$, $N_2 = N-e_2$. It is clear that there is a tree path from $s$ to $\ell$ in both $N_1$ and $N_2$.

For $i = 1,2$, let $t_i$ be the highest ancestor of $\ell$ in $G$ such that $N_i[s]$ displays $G[t_i]$. Without loss of generality, assume that $t_1$ is above $t_2$. Then the following proposition holds.

\begin{proposition}
$N$ displays $G$ if and only if $N_1$ displays $G$.
\end{proposition}
\begin{proof}
The sufficiency condition is trivial, as $N_1$ is a subnetwork of $N$.

To prove the necessity, we assume $T$ is a subtree of $N$ that is a subdivision of $G$ and contain no dummy leaf. If $T$ does not contain $e_1$, then $T$ is a subtree of $N_1$ and we are done. 
Otherwise, $T$ is a subtree of $N_2$, and so by the fact that $N_2$ can display at most $G[t_2]$, we have $T[s]$ is at most a subdivision of $G[t_2]$.

Let $T'$ be a subtree of $N_1[s]$ that is rooted at $s$ and is a subdivision of $G[t_1]$. The existence of $T'$ is guaranteed from the assumption that $N_1[s]$ displays $G[t_1]$. 
If $t_1 = t_2$, then the tree 
$$T - \mathcal{E}(T[s]) + \mathcal{E}(T')$$ 
is a subdivision of $G$ that is in $N_1$, and we are done.
Otherwise, $t_1$ is strictly above $t_2$. If $s_1$ is a node in $T$ that correspond to $t_1$ in $G$, then $s_1$ is strictly above $s$ as $N_2[s]$ does not display $G[t_1]$ and $T$ is a subtree of $N_2[s]$. 
Thus, we can consider the tree
$$T - \mathcal{E}(T[s_1]) + P[s_1,s] + \mathcal{E}(T'),$$
where $P[s_1,s]$ is the path from $s_1$ to $s$ in $T$.
The new tree is then a subtree of $N_1$ that is a subdivision of $G$, which completes the proof.
\end{proof}

Using the above proposition, we can dissolve $N[s]$ simply by calling {\sc Dissolve\_C2} twice as follows.

\begin{center}
\footnotesize
\begin{tabular}{l}
\hline
{\sc [$N'$, $G'$] = Dissolve\_C3 ($N$, $s$, $G$) }\\
\\
{\bf Input}: A binary phylogenetic tree $G$ and a reticulation-visible network $N$.  \\
\hspace*{1em} The node $s$ satisfies properties I, II, III, and C3.\\
{\bf Output}: "NO" if $N$ does not display $G$, otherwise output  $N'$ and $G'$ as in \Cref{prop_c2}.\\ 
1. Choose a leaf $\ell \in L_2(s)$;\\
\hspace*{1em} Let $e_1, e_2$ be the incoming edges of the reticulation parent of $\ell$.\\
2. Compute [$N_1, G_1, l_1$] = {\sc Dissolve\_C2($N-e_1, s, G$)};\\
3. Compute [$N_2, G_2, l_2$] = {\sc Dissolve\_C2($N-e_2, s, G$)};\\
4. Let $z = 1$ if $l_1 \geq l_2$ and let $z = 2$ otherwise;\\
5. Set $N' = N_z$ and $G' = G_z$.

\\
\hline

\end{tabular}
\end{center}

\section{Solving the tree containment problem}

Let $C_j$ be the lowest component of $N$ and let $\rho_j$ be its root.
Every leaf $\ell \in L_2(\rho_j)$ has a reticulation parent $r$, and both parents of $r$, say $u_1$ and $u_2$, are in $C_j$. 
We define $\mbox{sp}(\ell)$ to be the lowest common ancestor of $u_1$ and $u_2$ in $C_j$ (such node is also known as the split node for the reticulation $r$). 
A node $v \in \mathcal{V}(C_j)$ is stable on $\ell$ if and only if $v$ is above $\mbox{sp}(\ell)$. 
Finally, we define $S$ to be the set of split nodes in $C_j$, i.e.
$$S = \{\mbox{sp}(\ell) : \ell \in L_2(\rho_j) \}.$$
A node $s \in S$ is a lowest node in $S$, if there is no $s'\in S$ that is strictly below $s$.

Assume that $s = sp(\ell)$ is a lowest node in $S$.
$s$ is a stable at $\ell$. Furthermore, $s$ has two children, as otherwise it contradicts the fact that $sp(\ell)$ is the lowest common ancestors of the parents of $r$ ($r$ is the parent of $\ell$). Therefore $s$ satisfies property I and II. Let $s', s''$ denote the children of $s$, then the following proposition proves that $s$ also satisfies property III.

\begin{proposition} \label{prop_treepath}
The subnetwork $N[s']$ and $N[s'']$ are simply trees.
Furthermore, $s'$ (resp. $s''$) is stable if and only if there is a tree path from $s'$ (resp. $s''$) to a leaf.
\end{proposition}
\begin{proof}
Suppose on the contrary, $N[s']$ contains a reticulation $r'$ and both its parents. Let $\ell'$ be the child of $r'$. Then $s'$ is above $\mbox{sp}(\ell')$ and is strictly below $s$, contradicting the fact that $s$ is a lowest node in $S$.

Next, if $s'$ is stable, it is the stable ancestor of a leaf in either $L_1(s)$ or $L_2(s)$. The latter is impossible as $N[s']$ is a tree, so the former must hold, which further implies that there is a tree path from $s'$ to a leaf. Conversely, if there is a tree path from $s'$ to a leaf, then we can immediately deduce that $v$ is stable.
\end{proof}

We then order the elements of $S$ as $s_1, s_2, \ldots, s_p$ in post-order, so that $s_1$ is a lowest node in $S$.
It is not hard to see that if $N[s_1]$ is contracted into a single leaf, then $s_2$ becomes the next lowest node in $S$, assuming it satisfies property II. We can then repeatedly run either {\sc Dissolve\_C1}, {\sc Dissolve\_C2}, or {\sc Dissolve\_C3}, depending on the topology of the subnetwork $N[s_i]$, for every $s_i \in S$ in ascending index order.
If $s_p$ is the root $\rho_j$ of $C_j$, then this process ends with $N[\rho_j]$ contracted into a single leaf.
Otherwise, then the process terminated with the subnetwork $N[\rho_j]$ satisfying $L_2(\rho_j) = \emptyset$. As $\rho_j$ is stable, it must then be stable on a type-1 leaf, and so there is a tree path from $\rho_j$ to a leaf. 
We then run either {\sc Dissolve\_C1} or {\sc Dissolve\_C2} on $N[\rho_j]$ to finally contract $N[\rho_j]$ into a leaf.

Finally, we present the algorithm for solving tree containment problem for a binary reticulation-visible network $N$ and a binary tree $G$.

\begin{center}
\footnotesize
\begin{tabular}{l}
\hline
{\sc TCPSolver ($N$, $G$) }\\
\\
{\bf Input}: A binary phylogenetic tree $G$ and a binary reticulation-visible network $N$. \\
{\bf Output}: "NO" if $N$ does not display $G$, otherwise "YES".\\
\\

1. Traverse the network $N$, and find the big tree components $C_0, C_1, C_2, \ldots, C_q$, such\\
\hspace*{1em} that $\rho_j$ (root of $C_i$) is below $\rho_i$ only if $j \geq i$. $C_0$ is the component whose root is $\rho_N$.\\
\hspace*{1em} Pre-process $G$ so enquiring lowest common ancestor of two nodes takes $O(1)$ time;\\
2. For ($i = q:0$) \{\\
\hspace*{1em} 2.1 Pre-process $C_i$ as in \citep{harel1984fast} and compute $L_2(\rho_i)$;\\
\hspace*{1em} 2.2 Compute $S = \{sp(\ell): \ell \in L_2(\rho_i) \}$, order the elements of $S$\\
\hspace*{2.5em}  as $s_1, s_2, \ldots, s_p$ in post-order; set $s_{p+1} = \rho_k$;\\
\hspace*{1em} 2.3 For ($j = 1: p+1$), \{\\
\hspace*{4em} If ($s_j$ is a leaf), break for;\\
\hspace*{4em} Traverse $N[s_j]$;\\
\hspace*{4em} If (C1 holds) \{Call Dissolve\_C1($N,s_j,G$)\};\\
\hspace*{4em} ~~~ElseIf (C2 holds) \{Call Dissolve\_C2($N,s_j,G$)\};\\
\hspace*{4em} ~~~Else \{Call Dissolve\_C3($N,s_j,G$)\};\\

\hspace*{4em} If (subalgorithm return "NO") \{stop and return "NO"\};\\
\hspace*{4em} ~~~Else \{update $N$ and $G$ and continue\}\\
\hspace*{4em}\} \% end inner for\\
\hspace*{1em}\} \% end outer for\\
3. Return "YES";

\\
\hline

\end{tabular}
\end{center}

\noindent 

We first pre-process the network $N$ to find all the big tree component, and the tree $G$ so enquiring lowest common ancestor of two nodes becomes easy (see \citep{harel1984fast}). The pre-processing of $G$ requires $O(|\mathcal{V}(G)|)$ time, and is also used  step 2.1  for the tree component $C_i$.
The correctness of the algorithm follows from the previous discussion.

\textbf{Time complexity.}

We note that during the algorithm, we traverse a subtree of $G$, say $G[x]$, whenever we need to check whether it is displayed in a subnetwork $N[s]$ or not. 
If $G[x]$ has more nodes than $N[s]$, then we can simply terminate the traversal on $G[x]$ and deduce that $N[s]$ does not display $G[x]$. This allows us to bound the time complexity with the number of nodes in $N[s]$.

First, we show that  {\sc IsSubtree($T,G$)} runs in O($|\mathcal{V}(T)|$) time as follows.
Step 1 and 2 takes $O(|\mathcal{V}(T)|)$ time. 
Checking whether  $\mathcal{L}(G) \subseteq \mathcal{L}(T)$ in step 3 can be done in $O(|\mathcal{L}(T)|)$ time (if $|\mathcal{L}(G)| > |\mathcal{L}(T)|$, we can directly deduce that $\mathcal{L}(G) \not \subseteq \mathcal{L}(T)$). 
We continue the algorithm only if $\mathcal{L}(G) \subseteq \mathcal{L}(T)$. 
Finding the binary subtree $T'$ takes $O(|\mathcal{L}(G)|)$ time (\citep{cole2000n}), and comparing two binary trees in step 5 can be done in  $O(|\mathcal{L}(G)|)$ time too, and hence 
{\sc IsSubtree($T,G$)} runs in O($|\mathcal{V}(T)|$) time

Second, we show that {\sc Dissolve\_C1($N,s,G$)} runs in $O(|\mathcal{V}(N[s])|)$ time. 
Here, we assume that $|\mathcal{V}(G[t])| \leq |\mathcal{V}(N[s])|$, as otherwise we can immediately deduce $N[s]$ does not display $G[t]$.
Step 1 can be done by traversing $N[s]$ in breadth-first search manner so it takes $O(|\mathcal{V}(N[s])|)$ time. Step 2 can be done in $O(1)$ time, since we have pre-process $G$. $O(|\mathcal{L}(N[s])|+|\mathcal{L}(G[t])|)$ time. Step 4 takes $O(|\mathcal{V}(N[s'])|)$ for calling {\sc IsSubtree} and updating the network $N$.
As step 5 is symmetric with step 4, it takes $O(|\mathcal{V}(N[s''])|)$ time. 
Step 6 takes $O(|\mathcal{V}(G[t])|)$ time. 
Hence, the total time needed is $O(|\mathcal{V}(N[s])|)$. 

Third, we show that {\sc Dissolve\_C2($N,s,G$)} runs in $O(|\mathcal{V}(N[s])|)$ time. 
We again assume that $|\mathcal{V}(G[t])| \leq |\mathcal{V}(N[s])|$.
Step 1 can be done by traversing $N[s]$ in post-order manner, which requires $O(|\mathcal{V}(N[s])|)$ time. Step 2 and 3 simply takes constant time.
Consider an iteration of step 4. 
Step 4.1 takes $O(|\mathcal{V}(N[u_\alpha']))|$. Step 4.2 requires constant time, as we only need to check constant number of nodes. Step 4.3 requires $O(|\mathcal{V}(N[u_\alpha'])|)$ by the discussion in the previous paragraph. Step 4.4 also requires $O(|\mathcal{V}(N[u_\alpha'])|)$, according to the discussion for {\sc IsSubtree} time complexity above. Hence, an iteration of step 4 takes linear time with respect to the nodes of subnetwork under consideration. Afterwards, the subnetwork is contracted into a single leaf, so we always consider different nodes in the next iteration (except perhaps on some leaves in $L_2(s)$ and their parents, which are counted at most twice). Hence step 4 requires $O(|\mathcal{V}(N[s])|)$ time in total. Hence the algorithm {\sc Dissolve\_C2($N,s,G$)} runs in $O(|\mathcal{V}(N[s])|)$ time.

As the algorithm {\sc Dissolve\_C3($N,s,G$)} calls {\sc Dissolve\_C2} only twice, it also runs in $O(|\mathcal{V}(N[s])|)$ time. 

Finally, we consider the algorithm {\sc TCPSolver}.
Step 1 can be done by a breadth-first search on $N$, and requires $O(|\mathcal{E}(N)|)$.
Now, we consider an iteration of step 2. 
Step 2.1 can be done in $O(|\mathcal{V}(C_i)|)$ time. 
Step 2.2 can be done by inquiring the lowest common ancestor of the parents of $\mbox{pr}_N(\ell)$ whenever $\ell \in L_2(\rho_i)$, and traverse $C_i$ once for the post-order labeling.
For an iteration of step 2.3, we call one of the three algorithms in the previous section, each of which requires $O(|\mathcal{V}(N[s_i])|)$. At the end of the iteration, the subnetwork $N[s_i]$ under consideration is contracted into a leaf, and thus step 2.3 requires at most $O(|\mathcal{V}(N[\rho_i])|)$ time.
Hence the total time needed for step 2 is $O(|\mathcal{V}(N)|)$.

\citet{RECOMB2015} proved that the number of nodes and edges in a binary reticulation-visible network with $n$ leaves is $O(n)$. We conclude this section with the following theorem:

\begin{theorem}
If $N$ is a binary reticulation-visible network and $G$ is a binary tree with $n$ leaves, then the tree containment problem for $N$ and $G$ can be solved in $O(n)$ time.
\end{theorem}

\section{Conclusion}

We obtain a linear time algorithm for solving TCP for binary reticulation-visible networks, by
utilizing the fact that nodes with special properties as in Section 4 have simple structure below them.
The method is not limited for reticulation-visible networks; it can also be applied to any binary network in general, as long as there are nodes satisfying the special properties.

\bibliographystyle{plainnat}

\bibliography{sample}

\end{document}